\documentclass[12pt, a4paper]{article}
\usepackage[centertags]{amsmath}
\usepackage{amsmath,color}
\usepackage{amsfonts}
\usepackage{amssymb}
\usepackage{amsthm}
\usepackage{newlfont}
\usepackage{enumerate}
\usepackage{geometry}
\usepackage{tikz}
\usepackage{graphicx}
\usepackage[T1]{fontenc}
\usepackage{tikz}
\usetikzlibrary{plotmarks}
\usepackage{subfig}
\usepackage{url}
\usepackage{fancyhdr}
\usepackage[T1]{fontenc}
\usepackage[utf8]{inputenc}
\usepackage{floatrow}
\usepackage{setspace}
\usepackage[section,verbose]{placeins}
\usepackage{verbatim, amsmath, amsfonts, amssymb, amsthm}
\usepackage{textcomp}
\usepackage{setspace}
\usepackage{cite}
\usepackage{caption} 
\usepackage{lipsum} 
\usepackage{fancyhdr}
\renewcommand{\headrulewidth}{0pt}

\newtheorem{lemma}{\color{green}Lemma}

\newtheorem{proposition}{\color{red}Proposition}

\numberwithin{equation}{section}

\title{Moduli Spaces of Lumps on Real Projective Space}
\author{\textit{Steffen Krusch\thanks{S.Krusch@kent.ac.uk}}\hspace{0.2cm}  and \textit{Abera A. Muhamed\thanks{A.A.Muhamed@kent.ac.uk}}\\
School of Mathematics, Statistics and Actuarial Sciences\\
University of Kent, Canterbury, UK}
\date{\today}

\begin{document}
\maketitle
\begin{abstract}
Harmonic maps that minimise the Dirichlet energy in their homotopy classes are known as lumps. Lump solutions on real projective space are explicitly given by rational maps subject to a certain symmetry requirement. This has consequences for the behaviour of lumps and their symmetries. 
An interesting feature is that the moduli space of charge three lumps is a $7$-dimensional manifold of cohomogeneity one which can be described as a one-parameter family of symmetry orbits of $D_2$ symmetric maps.
In this paper, we discuss the charge three moduli spaces of lumps from two perspectives: discrete symmetries of lumps and the Riemann-Hurwitz formula. We then calculate the metric and find explicit formulas for various geometric quantities. We also discuss the implications for lump decay.
\end{abstract}

\renewcommand{\headrulewidth}{0pt}

\newcommand{\sectiontitle}[3][Your proposal title]{
\pagenumbering{arabic}
\def\sectionname{#3}
\begin{center}
\Large{\textbf{#1}}\\*[3mm] \Large{\textbf{\textsc{#2}}}
\end{center}
}

\section{Introduction}

Rational maps of degree $N$ are solutions of the Bogomolny equations of the $O(3)$ sigma model with topological charge $N$ and energy $2\pi |N|$. Belavin and Polyakov \cite{polyakov} studied the Bogomolny equations by change of variables and explored the Lagrangian density of the classically equivalent $\mathbb CP^1$ sigma model. The algebraic topology of rational maps and the construction of harmonic maps between surfaces have been studied by Segal \cite{segal} and by Eells and Lemaire \cite{eells}, respectively.  Speight and Sadun \cite{sadun} showed the moduli space for a compact Riemann surface is geodesically incomplete. The metric on the space of holomorphic maps is given by restricting the kinetic energy where the moduli space coordinates are allowed to depend on time. This metric is K\"{a}hler \cite{ruback}.
The low energy dynamics of a $\mathbb{C}P^1$ lump on the space-time $S^2\times\mathbb{R}$ \cite{speight1} and the geometry  of a space of rational maps of degree $N$ \cite{speight4} have been studied. The  Fubini-Study metric $\gamma_{FS}$ of rational maps of degree one has been studied by Krusch and Speight \cite{kruschspeight}. Here a rational map was identified with the projective equivalence classes of its coefficient such that $Rat_1$ is an open subset of $\mathbb{C}P^3$ which is equipped with the Fubini-Study metric of constant holomorphic sectional curvature $4$.
Lumps can decay  but it has been shown in Refs. \cite{ward,leese} that the scattering of lumps takes place before lump decay using the geodesic approximation. A head-on collision between lumps in the $2+1$-dimensional $\mathbb{C}P^1$ model on a flat torus has been studied numerically  by Cova and Zakrzewski \cite{cova} and analytically by Speight \cite{speightlump}.
Rational maps also play an important role in related models like the Skyrme model \cite{skyrme}. For example, the rational map ansatz \cite{houghton} gives a good approximations for the symmetries of Skyrme configurations, and the Finkelstein-Rubinstein constraints can be calculated directly from this ansatz using homotopy theory \cite{krusch-rat}. We study  the symmetries of rational maps to understand the geometry and dynamics of lumps. 

Harmonic maps are solutions of Laplace's equation on Riemannian manifolds and are  usually known as  lumps. 
Denote by $M_N$ the moduli space of degree $N$ harmonic maps.  $M_N$ is a  $(2N+1)$-dimensional smooth complex Riemannian manifold. There is a natural Riemannian metric on $M_N$, which is called the $L^2$ metric.  The $L^2$ metric  is well defined on a compact Riemann surface as the non-normalisable  zero modes are absent \cite{speight4}.
For the $\mathbb CP^1$ model, one can  have an explicit expression of harmonic maps in terms of rational maps  given by the ratio of two polynomials with no common roots. 
Lumps and their symmetries can  be understood in terms of rational maps  on the projective plane.
Speight \cite{speight2} studied the $L^2$ metric on  the moduli  spaces of degree 1 harmonic maps on both $S^2$ and $\mathbb RP^2$ and obtained an explicit formula. We  focus mainly on charge three rational maps between $\mathbb RP^2$, acquiring a  detailed and careful understanding of their $L^2$ geometry. 
The $L^2$ metric plays an important role in slow lump dynamics just as Samols' metric does for vortices \cite{samols}.

In section \ref{intro}, we discuss the $O(3)$-sigma model and the equivalent $\mathbb{C}P^1$-sigma model on a Riemann surface.  In section \ref{Sym}, we derive families of symmetric rational maps in real projective space. We then apply the Riemann-Hurwitz formula. We also study an $SO(3) \times SO(3)$ invariant angular integral of rational maps which plays an important role in our understanding of the moduli space of charge three lumps. This leads to the proof presented in appendix \ref{Appendix} that all charge three lumps have dihedral symmetry $D_2.$
In section \ref{Rat3}, we  discuss the moduli space metric of charge three lumps on projective space. We evaluate the metric coefficients explicitly and calculate various geometric quantities. We end with a conclusion.

\section{The $O(3)$-sigma model on a Riemann surface}
\label{intro}

In $(d+1)$-dimensional space time, the non-linear  sigma model on the space $\Sigma$ with target space $Y$  is defined by the Lagrangian
\begin{equation}
\label{lagr0}
L=\frac{1}{2}\int_{\Sigma}d\mu_g\partial_{\mu}\phi_l \partial^{\mu}\phi_m H^{lm},
\end{equation}
where $d\mu_g$ is the volume element of $\Sigma$, $g$ is the Riemannian metric on $\Sigma$, $\partial^{\mu}=\eta^{\mu\nu}\partial_{\nu}$ and $\eta^{\mu\nu}$ are the components of the inverse of the Lorentzian metric
$$
\eta=dt^2-g
$$ 
on the space-time $\mathbb{R}\times\Sigma,$ and $H_{lm}$ is the metric on $Y.$   
The $O(3)$-sigma model is a famous example of a  non-linear sigma model.   In the $O(3)$-sigma model, the field can be parameterised as a three-component unit
vector, $\phi=(\phi_1,\phi_2,\phi_3)$ with $\phi\cdot\phi=1$, and the Lagrangian is given by 
\begin{equation}
\label{lagr1}
L=\frac{1}{4}\int_{\Sigma}d\mu_g\partial_{\mu}\phi\cdot\partial^{\mu}\phi.
\end{equation}
Thus, the target space can be identified with a Riemann sphere $S^2$. The sigma model can be formulated in terms of fields  $(\phi_1,\phi_2,\sigma)$ such that   $\sigma=\pm\sqrt{1-(\phi_1^2+\phi_2^2)}$ where $\phi_1$ and $\phi_2$  are locally unconstrained \cite{manton}.  We  consider $\Sigma = S^2$ and $\Sigma = {\mathbb R}P^2,$ and target space $Y=\Sigma$, but other examples have been discussed in the literature \cite{speight2}.
Denote by $M_N$ the moduli space of degree $N$ static solutions of the $O(3)$ model on $\Sigma.$  The kinetic energy functional induces a natural metric $\gamma$ on $M_N$ which is a finite dimensional, smooth Riemannian manifold. The homotopy classes of a continuous map $\phi$, by the Hopf degree theorem \cite{bott}, are labelled by the topological degree of $\phi$. 

In the following we consider $\Sigma=S^2$. Let  $R$ be the stereographic 
coordinate image of $\phi$  on the target space. The coordinate $R$  is given by 
$R=\frac{\phi_1+i\phi_2}{1+\phi_3}$ and let the local complex
coordinate be $z=x^1+i x^2$ with complex conjugate $\bar{z}=x^1-i x^2$. One can then explicitly express $\phi$ in terms of $R$ as
\begin{equation*}
\phi=\left(\frac{R+\bar{R}}{1+|R|^2}, 
\frac{-i\left(R-\bar{R}\right)}{1+|R|^2}, 
\frac{1-|R|^2}{1+|R|^2}\right).
\end{equation*}
Since $R=R(t, z,\bar{z})$ is a function of $t$, $z$ and $\bar{z}$ , the  Lagrangian (\ref{lagr1}) becomes
\begin{equation}
L=\int_{S^2}dS\frac{\partial_{\mu}
R\partial^{\mu}\bar{R}}{(1+|R|^2)^2},
\end{equation}
where $\mu=t, z,\bar{z}$. This Lagrangian \eqref{lagr1} is referred to as the $\mathbb{C}P^1$ sigma model, and it is equivalent to the $O(3)$ sigma model. The $\mathbb{C}P^1$ sigma model in $d=2+1$ dimensions is a non-linear
field theory possessing topological solitons, called lumps. 

For $\Sigma=S^2$, the energy $E$ and the topological charge $N$ are given by
\begin{align}
&E=2\int\frac{(|\partial_z R|^2
+|\partial_{\bar{z}}R|^2)(1+|z|^2)^2}{(1+|R|^2)^2}
\frac{dzd\bar{z}}{(1+|z|^2)^2},\\ 
&N=\frac{1}{\pi}\int\frac{(|\partial_z R|^2
-|\partial_{\bar{z}}R|^2)(1+|z|^2)^2}{(1+|R|^2)^2}
\frac{dzd\bar{z}}{(1+|z|^2)^2},
\end{align}
where $\partial_z = \frac{\partial}{\partial z} = \frac{1}{2}(\partial_1-i\partial_2)$ and $\partial_{\bar{z}} = \frac{\partial}{\partial\bar{z}} = \frac{1}{2}(\partial_1+i\partial_2)$.
Now one can show that the energy $E$ and the topological degree $N$ satisfy the Bogomolny bound $E\geq2\pi N$.  Equality holds if and only if the Cauchy-Riemann
equation is satisfied, namely $\partial_{\bar{z}}R=0$,
whose solutions are holomorphic functions $R(z)$. We can do a similar calculation to obtain $E \geq - 2\pi N$ with equality when $\partial_z R=0$,  which is satisfied by antiholomorphic functions $R(\bar{z})$.
In summary, the energy $E$ is minimised to $2\pi|N|$ in each
topological class by a solution of the Cauchy-Riemann equations
\begin{equation*}
\begin{cases}
\partial_{\bar{z}}R=0& \text{if } N\geq 0,\\
\partial_zR=0& \text{if } N\leq0.
\end{cases}
\end{equation*}
Without loss of generality we will focus on holomorphic maps. 
For the complex coordinates $z$ and $R$ on the domain and codomain, the general degree $N$ rational map is 
\begin{equation*}
R(z)=\frac{p(z)}{q(z)}=
\frac{a_1+a_2z+...+a_{N+1}z^N}{a_{N+2}+a_{N+3}z+...+a_{2N+2}z^N},
\end{equation*}
where $a_i\in\mathbb{C}$ are constants and $a_{N+1}$ and $a_{2N+2}$
do not both vanish simultaneously,  and $p(z)$ and $q(z)$ have no common
roots. Suppose $a_{2N+2}\neq 0$ and define a complex coordinate
$b_{\alpha}=\frac{a_{\alpha}}{a_{2N+2}}, \alpha=1,...,2N+1$. The inclusion property $M_N\subset\mathbb{C}P^{2N+1}$ implies that the metric $\gamma$ is
K\"{a}hler in this coordinate system. Thus, $\gamma$ is given by
\begin{equation}
\gamma=\gamma_{\alpha\beta}db^{\alpha}\overline{db^{\beta}}.
\end{equation}
The metric $\gamma_{\alpha\beta}$ can be written as 
\begin{equation}
%\gamma_{\alpha\beta}=\int_{\mathbb{C}}\frac{dzd\bar{z}}{(1+|z^2|)^2(1+|R|^2)^2}% \frac{\partial
%R}{\partial b^{\alpha}} \overline{\frac{\partial R}{\partial
%b^{\beta}}},
\gamma_{\alpha\beta}=\int_{\mathbb{C}}
\frac{1}{(1+|R|^2)^2}
\frac{\partial R}{\partial b^{\alpha}} \overline{\frac{\partial R}{\partial
b^{\beta}}}
\frac{dzd\bar{z}}{(1+|z|^2)^2},
\end{equation}
and $R(z)$ is given by
\begin{equation*}
R(z)=\frac{b_1+b_2z...+b_{N+1}z^N}{b_{N+2}+b_{N+3}z+...+z^N}.
\end{equation*}  

\section{Symmetries of rational maps on projective space}
\label{Sym}

In the following, we derive families of symmetric rational maps in real projective space. We start by considering rational maps between Riemann spheres. A rational map $R(z)$ has a discrete symmetry if 
\begin{equation}
\label{sym}
M_1\left(R\left(z\right)\right) = R\left(M_2\left(z\right)\right) ,
\end{equation}
where $M_1$ and $M_2$ are M\"obius transformations. $M_1$ is a rotation in target space, also  known as an isorotation, whereas $M_2$ is a rotation in space. Note if $R(z)$ has symmetry (\ref{sym}) then
the rational map ${\tilde R}(z) = {\hat M}_1^{-1} (R({\hat M}_2(z)))$ has the symmetry
$$
{\tilde M}_1 \left({\tilde R}(z)\right) = {\tilde R} \left({\tilde M}_2 (z)\right),
$$
where ${\tilde M}_1(z)  = {\hat M}_1^{-1} \left(M_1 \left({\hat M_1}(z)\right)\right)$ and ${\tilde M}_2 = {\hat M}_2^{-1}\left( M_2\left( {\hat M}_2 (z)\right)\right).$ So, by change of coordinates in domain and target, we can choose our symmetry to be around convenient axes.

We define a $C_n^k$ symmetry of a rational map as a rotation in space by $\alpha = 2\pi/n$ followed by a rotation in target space by $\beta = k\alpha.$ 
The following lemma classifies which $C_n^k$ symmetries are allowed for a rational map of degree $N.$

\begin{lemma}
\label{lem1}
A rational map of degree $N$ can have a $C_n^k$ symmetry if and only if
$N \equiv 0 \mod n$ or $N \equiv k \mod n$.
\end{lemma}

For a proof see \cite{kruschfinkelstein}. 
Note that    $S^2$ is the universal covering space of $\mathbb{R}P^2$ and $\pi_1(\mathbb{R}P^2)=\mathbb{Z}_2$ \cite{massey}.  The map
$\phi:S^2\to S^2$ projects to a harmonic  map
$\tilde{\phi}:\mathbb{R}P^2\to\mathbb{R}P^2$
 if and only if $p\circ\phi=\phi\circ p$ \cite{eells,speight2}, where  
$p:S^2\longrightarrow S^2$, $z\mapsto
-\bar{z}^{-1}$ is the antipodal map  in stereographic coordinates.
Hence, in real projective space,  rational maps must satisfy the additional constraint
\begin{equation}
\label{constraint}
R\left(-\frac{1}{{\bar z}}\right) = -\frac{1}{{\bar R(z)}}.
\end{equation}
 Then $R(z)$ can be written in the general form as
\begin{equation}
R(z) = \frac{\sum\limits_{k=0}^N a_k z^k}
{\sum\limits_{k=0}^N (-1)^k {\bar a}_{N-k} z^k},
\end{equation}
where $a_k \in {\mathbb C}.$ Alternatively, we can write the general form of this rational map as
\begin{equation}
\label{degn}
R(z)=e^{i\phi}\frac{(z-z_1)(z-z_2)...(z-z_N)}{(1+\bar{z}_1z)(1+\bar{z}_2z)...(1+\bar{z}_Nz)},
\end{equation}
where $z_1,z_2,...,z_n\in\mathbb{C},$ and $\phi\in[0,2\pi).$
This severely restricts the number of allowed symmetries. Denote by $\widetilde{Rat_N}$ the degree $N$ rational maps in real projective space  which is a submanifold of $Rat_N$.
\begin{lemma}
\label{lem2}
A rational map of degree $N$ which satisfies (\ref{constraint}) can have a $C_n^k$ symmetry if and only if
$N \equiv k \mod n$.  
If $n \ge N$ then the rational map has $D_\infty$ symmetry.
\end{lemma}

\noindent {\bf Proof:}\\
Without loss of generality we choose coordinates such that one $C_n^k$ rotation is around the $x_3$-axis in space and target space and satisfies the boundary condition $R(\infty) \neq 0.$  

First consider $N \equiv 0 \mod n.$ Then $N=nl,$ and for $0<k<n$ a $C_n^k$ rational map can be written as 
$$
R(z) = \frac{r(z^n)}{z^{n-k}s(z^n)},
$$
where $r(z)$ has degree $l$ and $s(z)$ has degree less than $l.$ 
On the other hand, given $r(z)$ the constraint (\ref{constraint}) fixes the coefficients of the denominator. In particular, only coefficients of  powers of $z^n$ will be non-zero. Hence, there are no solutions with $0<k<n,$ and the only possibility is 
$$
R(z) = \frac{r(z^n)}{s(z^n)},
$$
where $r(z)$ has degree $l,$ and the degree of $s(z)$ is less or equal to $l.$

Consider the case  $N \equiv k \mod n$ which includes the $k=0$ case for $N \equiv 0 \mod n.$ Set $k = N \mod n$ and $s = (N-k)/n,$ then the rational map is given by
\begin{equation}
\label{Ratsym1}
R(z) = \frac{\sum\limits_{j=0}^s a_j z^{j n + k}}{\sum\limits_{j=0}^s b_{j} z^{jn}}.
\end{equation} 
The inversion symmetry (\ref{constraint}) leads to the following two constraints on the coefficients
\begin{equation}
\label{c1}
(-1)^{nj}{\bar b}_{s-j} = \lambda a_j  
\end{equation}
and
\begin{equation}
\label{c2}
(-1)^{k+1} {\bar a}_{s-j} (-1)^{nj}  = \lambda {\bar b}_j,
\end{equation}
where $\lambda$ takes account of the fact that numerator and denominator can be multiplied with a common factor.
Taking the modulus, we obtain that $|\lambda|=1,$ so that ${\bar \lambda}=1/\lambda.$ By relabelling $j \mapsto s-j,$ equation (\ref{c2}) becomes
\begin{equation}
a_j = {\bar \lambda} (-1)^{k+1+n(s-j)} {\bar b}_{s-j}.
\end{equation}
This is compatible with equation (\ref{c1}) provided $ns+k = N$ is odd. 
For $n=N,$ we obtain the map
\begin{equation}
\label{Dinf}
R(z) = \lambda \frac{a_1 z^N+ a_0}{-{\bar a}_0 z^N + {\bar a_1}}.
\end{equation}
Performing a M\"obius transformation in target space to remove the phase $\lambda$ this is equivalent to a M\"obius transformation of the axial map. 

Similarly, for $n>N,$ the rational map (\ref{Ratsym1}) reduces to
$$
R(z) = \frac{\lambda a_0}{\bar a_0} z^N,  
$$
since in this case $N=k$ and $s=0.$ This is again the axial map. \hfill$\square$

For $N=1,$ the only rational maps compatible with the constraint (\ref{constraint}) is
\begin{equation}
R(z) = \frac{a z+ b}{-{\bar b} z + {\bar a}},
\end{equation}
where $|a|^2+|b|^2 \neq 1.$ Hence, this is isomorphic to $PU(2) \cong SO(3),$ and the moduli space for charge one is $SO(3),$ as pointed out in Ref. \cite{speight2}. 

\subsection{$N=3$ lumps}

In the following, we will discuss the case $N=3$ in more detail. According to Lemma \ref{lem2}, imposing $C_n$ symmetry with $n\ge 3,$ we obtain $D_\infty$ symmetry, given by maps of the form
\begin{equation}
\label{Rinf}
R(z) =  \frac{a z^3 + b}{-{\bar b} z^3 +  {\bar a}},
\end{equation}
where $a, b \in {\mathbb C},$ and $a$ and $b$ are not both zero. Here the rotation axis in space has been chosen to be the $x_3$-axis. This choice corresponds to fixing two real parameters.  Here, the symmetry is a $C_3^0.$ Consider the case $a\neq0$. Then the rational map (\ref{Rinf})  can be rewritten as 
\begin{align*}
R(z) &
={\mathrm e}^{i\psi} \frac{z^3 + c}{-\bar c z^3 +  1},
\end{align*}
where $c=b/a$ and $e^{i\psi}=a/{\bar a}$.
Hence, the moduli space of the symmetric lumps of (\ref{Rinf}) is parametrised by one complex number and a phase which together with the choice of axes gives real dimension $5.$ The moduli space can also be viewed as the orbit under rotations and isorotations of the map 
$$
R(z)=z^3.
$$
Since, rotation and isorotations act independently apart from the axial symmetry around the third axis, the dimension of the moduli space is again $5.$

The only rational maps that are compatible with a $C_2$ symmetry around the $x_3$-axis and satisfy $R(\infty) \neq 0$ are given by
\begin{equation}
\label{D2}
R(z) = \frac{a z^3 + bz}{{\bar b} z^2 + {\bar a}}.
\end{equation}
By M\"obius transformations preserving this symmetry, namely rotations around the third axis in space and target space, the rational map can be brought into the form
\begin{equation}
\label{Rc}
R(z) = \frac{z^3 + cz}{c z^2 + 1},
\end{equation}
such that $c$ is real and non-negative. The surprising fact is that this map has $D_2$ symmetry, since it is also symmetric under
\begin{equation}
\label{C2}
R\left(-\frac{1}{z}\right) = -\frac{1}{R\left(z\right)}.
\end{equation}
Hence, imposing a $C_2$ symmetry automatically gives a family of $D_2$ symmetric maps. Since rotations and isorotations act independently and cannot change the magnitude of the parameter $c,$ the moduli space of the symmetric lumps of (\ref{Rc}) has real dimension $7.$ Another way of calculating the dimension of the family of symmetry orbits of (\ref{Rc}) is as follows. A general rational map can be written as equation (\ref{degn}) for $N=3$:
\begin{equation}
\label{deg3}
R(z) =  {\mathrm e}^{i\phi} \frac{(z-z_1)(z-z_2)(z-z_3)}
{(1+{\bar z}_1 z)(1+{\bar z}_2 z)(1+{\bar z}_3 z)}.
\end{equation}
When we impose symmetry under a $\pi$ rotation around the $x_3$-axes in space followed by an isorotation around the $x_3$-axis in target space, one zero has to be equal to zero and the other two map into each other under $z \mapsto -z$. The symmetric rational map is then given by 
\begin{equation}
R(z) =  {\mathrm e}^{i\phi} \frac{z(z-z_1)(z+z_1)}
{(1+{\bar z}_1 z)(1-{\bar z}_1 z)}.
\end{equation}
Hence, this rational map is parametrised by $\phi \in {\mathbb R}$ and $z_1 \in {\mathbb C},$ that is by $3$ real parameters. A further $4$ real parameters correspond to our choice of axes on $S^2\times S^2.$ 

To list all  families of degree $3$ rational maps on the projective plane, we can also use the Riemann-Hurwitz formula \cite{fulton1}. For a degree $N$ rational map $R(z)$ ramified at points $p_i$ in $S^2$, the Riemann-Hurwitz formula is given by 
\begin{equation}
\label{RR}
\chi(S^2)=N\chi(S^2)-\sum_{p_i}\left(d_{p_i}-1\right),
\end{equation}
where $\chi(S^2)$ is the Euler characteristic of $S^2,$ and $d_{p_i}$ is the ramification index of $R(z)$ at the point $p_i$.
Thus, for $N=3$, \, $\sum_{p_i}\left(d_{p_i}-1\right)=4$ implies that there are two possibilities. The first possibility is $d_{p_1}=3$ which fixes $d_{p_2}=3$ using symmetry \eqref{constraint}.
Hence, the first family is the symmetry orbit of rational maps $z^3$ which coincides with the family of  maps (\ref{Rinf}). The second possibility is  $d_{p_i}=2,$ for $i=1,2,3,4$. Choosing a critical point at $z=0$ with $R(0)=0$ and then using M\"obius transformations, we can find the family of rational maps as
\begin{equation}
\label{Ra}
R(z) = \frac{z^2(z-a)}{1+az}, \quad a>0.
\end{equation}
Note that the sign of $a$ can be changed by $R(z) \mapsto -R(-z)$.  Here rotations and isorotations act independently and cannot change the magnitude of $a,$ hence the space is $7$-dimensional. This is compatible with writing the rational map as equation (\ref{deg3}).

A useful angular integral of a degree $N$ rational map $R$ is given by 
\begin{equation} 
\label{calI} 
{\cal I}
=\frac{1}{4\pi}\int\left(\frac{1+|z|^2}{1+|R|^2}\left|\frac{dR}{dz}\right|\right)^4
\frac{dzd\bar{z}}{(1+|z|^2)^2}.
\end{equation}
It can be shown that $\cal I\geq N^2$ \cite{manton}.  This quantity ${\cal I}$ is invariant under rotations in space and rotations in target space. Therefore, it distinguishes between maps not related by the symmetry group. In Fig.~\ref{RaRcinertia}, we display the angular integral ${\cal I}$ for the rational map $R_c$ in (\ref{Rc}) as a function of $c.$
For $c=0,$ we obtain the axial map $R(z) = z^3,$ for which the angular integral
${\cal I}$ can be evaluated explicitly to $\cal I_0=\frac{81+16\sqrt{3}\pi}{9}.$ As $c \to 1,$ the integral ${\cal I}$ diverges as two lumps become spiky, see Fig.~\ref{enrdsyRc}, and for $c=1$ the rational map becomes $R(z) = z.$ For $c=3,$ the angular integral ${\cal I}$ again takes the value ${\cal I}_0.$ This can be understood as follows: Since the $D_2$ symmetry fixes the three axes $x_1,$ $x_2$ and $x_3,$ we can permute the  $(x_1,x_2,x_3)$ axes cyclically and obtain, after a suitable isorotation, another map of the form (\ref{Rc}), but with a new value ${\tilde c} = \frac{3-c}{c+1}.$ This relates $c=0$ to ${\tilde c} = 3,$ which is the axial map in disguise. 
Furthermore, the interval $0 \le c < 1$ is mapped to $1 < {\tilde c} \le 3.$ Similarly, we can map the interval $0 \le c < 1$ to $3 \le {\hat c} < \infty,$ where ${\hat c} = \frac{ c+3}{ 1-c}.$ Thus, the angular integral ${\cal I}$ of (\ref{Rc}) for  $0<c<1$  and  $1<{\tilde c}<3$ are identical, and similarly for $3<{\hat c}<\infty$. Hence, it is sufficient to consider the interval $0 \le c < 1,$ and its symmetry orbits in order to parametrise $\widetilde{Rat_3}.$ 

\begin{figure}[h!]
\centering
\includegraphics[width=0.765\textwidth]{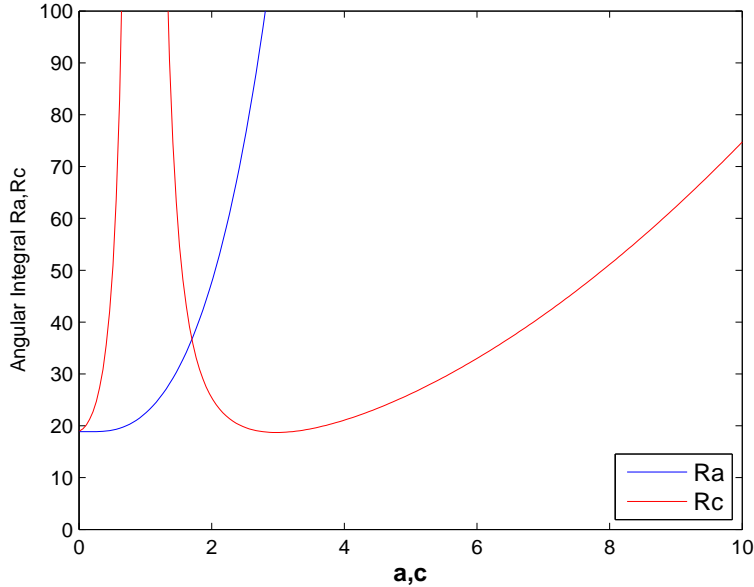}
\caption{Plots of the angular integral ${\cal I}$ in (\ref{calI}) as a 
function of $a$ and $c$ for $R_a$ and $R_c,$ respectively.}
\label{RaRcinertia}
\end{figure}

In appendix \ref{Appendix} we have constructed the M\"obius transformations in space and target space that are needed to transform the rational map $R_c$ in (\ref{Rc}) into  the rational map $R_a$ in (\ref{Ra}) for $0<c<1.$ The key idea of the calculation is to first map the Wronskians of $R_a$ and $R_c$ into each other. Surprisingly, this leads to the simple equation \eqref{a(c)} for $a$ as a function of $c,$ given by 
$$
a = \sqrt{\frac{c(c+3)}{1-c}}.
$$
Equation \eqref{a(c)} maps $c \in [0,1)$ to $a \in [0, \infty),$ and we checked explicitly that the blue and red curves in figure \ref{RaRcinertia} can be made to lie on top of each other within numerical accuracy.
Hence, the moduli space of lumps of (\ref{Rc}) and (\ref{Ra}) parametrise 
the same space, and therefore (\ref{Ra}) has a hidden $D_2$ symmetry.

The maps (\ref{Ra}) and (\ref{Rc}) can also be used to discuss interesting 
lump decay channels. First, consider maps $R_a \in \widetilde{Rat_3}$ of the 
form (\ref{Ra}). Following the zeros and poles for $a \in [0,\infty),$ we 
start with the axial map with three zeros at the origin $0$ and three 
poles at $\infty$. Then one zero moves from $0$ to $\infty$ along the 
positive real axis while one pole moves from $\infty$ to $0$ along the 
negative real axis. The zero cancels with a pole at $\infty$ while the pole 
cancels with a zero at $0.$

Second, consider maps $R_c \in \widetilde{Rat_3}$ of the form (\ref{Rc}), as 
displayed in Fig.~\ref{enrdsyRc}:  For $c=0$ we obtain the axial maps with 
three zeros at the origin and three poles at $\infty.$ For $0<c<1$
one zero remains fixed while one zero moves up and one moves down along 
the imaginary axis. Also, one pole remains fixed at $\infty$
while two pole travel towards $0$ along the positive and negative 
imaginary axis, respectively. For $c=1$ two poles and two zeros cancel. 
For $1 < c <\infty$ the poles move towards $0$ while the poles move 
towards $\infty$ where they cancel.

\subsection{$N \ge 5$ lumps}

In the following we discuss the symmetries of $\widetilde{Rat_N}$ for $N\ge 5.$
We first discuss the symmetry approach. Then we compare the $N=5$ case with the results of the Riemann-Hurwitz formula. 

According to Lemma \ref{lem2}, when a $C_n$ symmetry is imposed with $n\ge N$ then the rational map has $D_{\infty}$ symmetry. 
The moduli space of all axially symmetric rational maps of degree $N$ is generated by rotations and isorotations of $R(z) = z^N$ and is
$5$-dimensional. Note that since a zero and a corresponding pole are opposite to each other, a single lump cannot decay, and the moduli space is $\widetilde{Rat_1} \cong SO(3).$ Similarly, the axial symmetry $D_\infty$ fixes all $N$ zeros and $N$ poles so that axial symmetry prevents lump decay.

If we impose the symmetry $C_n^k$ where $N=n+k$ and $0<k<n$ then the family of orbits of such maps is a $7$-dimensional space as the following proposition shows.
\begin{proposition}
\label{dnsym}
If $N=n+k$ with $0<k<n,$ then the moduli space of rational maps with $C_n^k$ symmetry is $7$-dimensional and has $D_n$ symmetry.
\end{proposition} 
\begin{proof}
A rational map of degree $N=n+k$ with $C_n^k$ symmetry is given by
\begin{equation}
\label{dnsym1}
R(z)=\frac{z^k(a_1z^n+a_0)}{(-1)^n\bar a_0z^n+\bar a_1}.
\end{equation}
By M\"obius transformation, the rational map can be brought into the form
\begin{equation}
\label{dnsym2}
R(z)=\frac{z^k(z^n+a)}{1+(-1)^naz^n},\, a\in[0,\infty).
\end{equation}
Furthermore, $a \neq 1$ for even $n.$ Since $a$ is real, 
the map (\ref{dnsym2}) satisfies the $C_2^1$ symmetry \eqref{C2}.
Hence, the moduli space of the rational map has $D_n$ symmetry. Since rotations and isorotations act independently on (\ref{dnsym2}) and do not change the value of $a,$ the moduli space is $7$-dimensional. 
\end{proof}
These symmetric moduli spaces are very similar to the moduli spaces of vortex polygons in \cite{kruschspeight2} and are also related to the cyclic monopole scattering in \cite{Hitchin:1995qw}. The Wronskian of \eqref{dnsym2} is given 
by
\begin{equation}
\label{W}
W(z) = z^{k-1} \left
(k (-1)^n az^{2n} + \left(n+k + (-1)^n (k - n)a^2\right) z^n +a k
\right).
\end{equation}
It has $k-1$ zeros at $0$ and $k-1$ at $\infty,$ as well as two circles around the origin that each contain $n$ zeros. The $C_n^k$ symmetry moves zeros on the two circles into each other keeping $0$ and $\infty$ fixed, whereas the $C_2^1$ symmetries map $0$ and $\infty$ and the two circles into each other. 

In the following, we list the symmetric submanifolds for $N=5.$ The axial map $R(z) = z^5$ has $D_\infty$ symmetry. According to proposition \ref{dnsym} there is a $D_4$ symmetric family of maps 
\begin{equation}
\label{rat55}
R(z)=\frac{z(z^4+a)}{1+az^4},\, a\in(0,1)\cup(1,\infty).
\end{equation}
If $a=-5$, then the map has an additional $C_3$ symmetry, the symmetry group is enhanced to octahedral symmetry $O$ \cite{manton} which gives rise to a $6$-dimensional subspace. 
Proposition \ref{dnsym} also leads to the $D_3$ symmetric map
\begin{equation}
\label{rat54}
R(z)=\frac{z^2(z^3+a)}{1-az^3},\, a\in[0,\infty),
\end{equation}

When imposing $C_2^1$ symmetry we obtain the following family of rational maps
\begin{equation}
\label{rat52}
R(z)=\frac{z(a_0+a_1z^2+a_2z^4)}{\bar a_2-\bar a_1z^2+\bar a_0z^4}.
\end{equation}
This is a $9$-dimensional space, and therefore a proper subspace of the $11$-dimensional space $\widetilde{Rat_5}.$ In general, there are no further symmetries. An example of a rational map which satisfies the $D_2$ symmetry is given by
\begin{equation}
\label{rat52eg}
R(z)=\frac{z(1+iaz^2+bz^4)}{b+iaz^2+z^4}, \, a,b\in\mathbb R.
\end{equation}

In the following, we briefly explore the Riemann-Hurwitz formula \cite{fulton1} for charge five lumps. Thus, for $N=5$, \, $\sum_{p_i}\left(d_{p_i}-1\right)=8$ implies that there are five possibilities.  
The first possibility is $d_{p_1} = 5$ which fixes $d_{p_2}=5$. Then we obtain the rational map $z^5$ which is the axially symmetry map. The second possibility is $d_{p_1}=d_{p_2}=4$ which fix $d_{p_3}=d_{p_4}=2$.  In this case, the  family of rational maps is given by
\begin{equation}
\label{rat511}
R(z)=\frac{z^4(z+a)}{1-az},
\end{equation}
for $0 < a < \infty.$ The Wronskian has $3$ zeros at $0,$ three zeros at $\infty$ and two simple zeros. It is clear that this space is quite different from the symmetric spaces discussed earlier. The remaining possibilities are $d_{p_1}=d_{p_2}=3$ and $d_{p_i}=2$ for $i=3,4,5,6,$ $d_{p_i}=3$ for $i=1,2,3,4$ and $d_{p_i}=2,\, i=1,...,8,$ which are discussed in more detail in \cite{Muhamed}.

In summary, for $N \ge 5,$ there are interesting $7$-dimensional spaces with $D_n$ symmetry. However, the Riemann-Hurwitz theorem produces completely different spaces, with the exception of the axially symmetric case.

\section{The moduli space $\widetilde{Rat_3}$ on $\mathbb{R}P^2$}
\label{Rat3}
In this section we discuss the moduli space of charge three lumps on real projective space. We calculate the metric and various geometric quantities. We first discuss maps of the form (\ref{Rc}) which possess dihedral symmetry $D_2$.
We describe the $7$-dimensional family of symmetry orbits of $R_c \in \widetilde{Rat_3}$ and maps of the form (\ref{Rinf}) which possess axial symmetry  $D_\infty$.

Consider first the space of functions $R_c$ given by (\ref{Rc}).
The Wronskian of $R_c$  is a polynomial of degree $4$ given by 
$$
w(z)=cz^4+(3-c^2)z^2+c.
$$
For $c\approx1$, the poles and zeros of $R_c$ come together and cancel each other, then $R_c$ becomes a rational map of degree one which is $z$. For $c\approx\infty$, the poles and zeros come together and cancel each other and then $R_c$ becomes the rational map $\frac{1}{z}$.  
One can see the energy density of this space in Fig.~\ref{enrdsyRc} that shows the energy density is symmetric at $c=0$, $c=1$ and  $c=\infty$. The energy densities dissociate and form spikes as $c$ approaches $1$ and $\infty.$

\begin{figure}[ht!]
\centering
\subfloat[Energy density I][$c=0$]{
\includegraphics[width=0.3\textwidth]{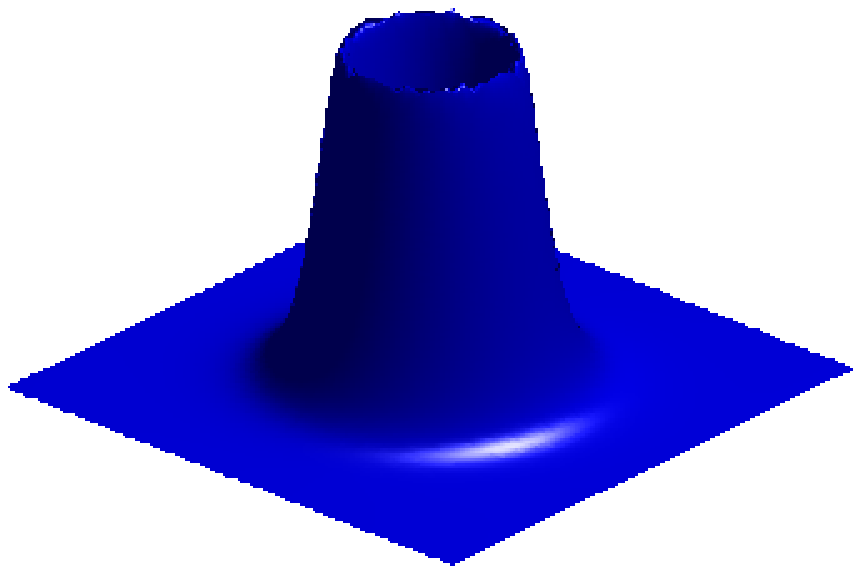}
\label{figa}}
\subfloat[Energy density I][$c=0.1$]{
\includegraphics[width=0.3\textwidth]{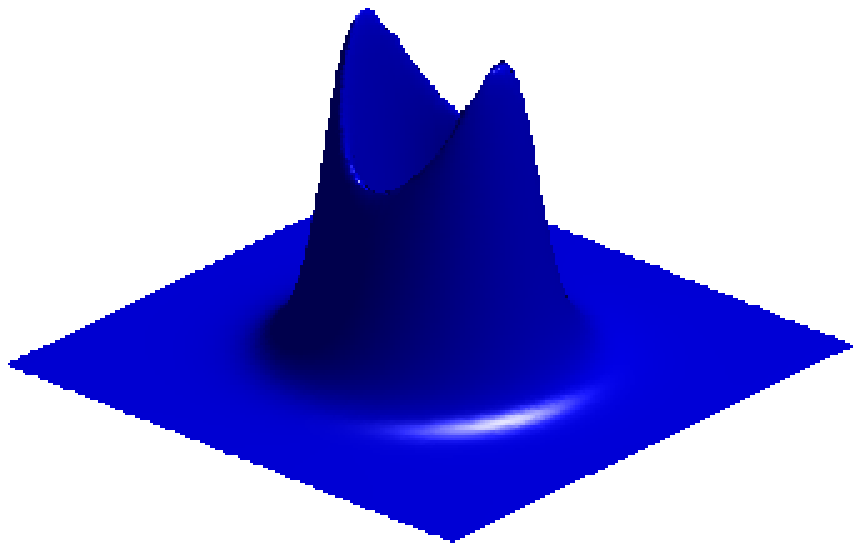}
\label{figb}}
\subfloat[Energy density I][$c=0.5$]{
\includegraphics[width=0.3\textwidth]{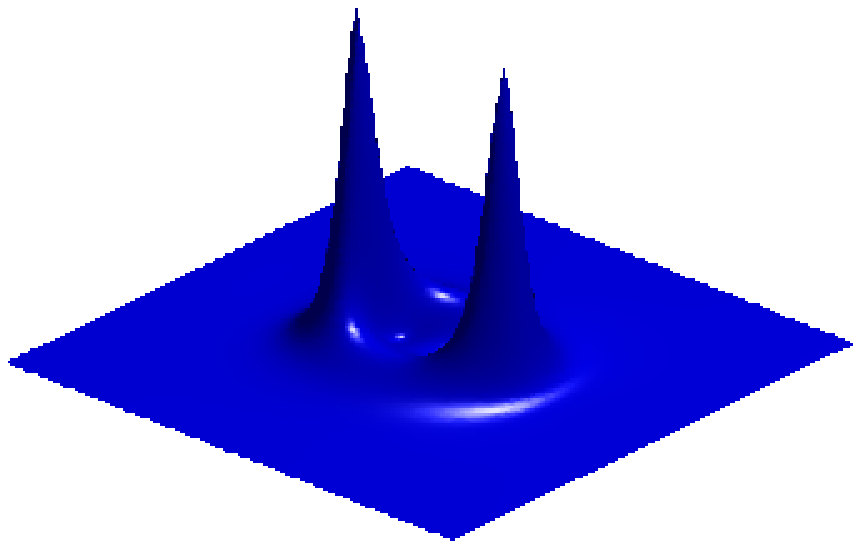}
\label{figc}}
\qquad
\subfloat[Energy density I][$c=0.9$]{
\includegraphics[width=0.3\textwidth]{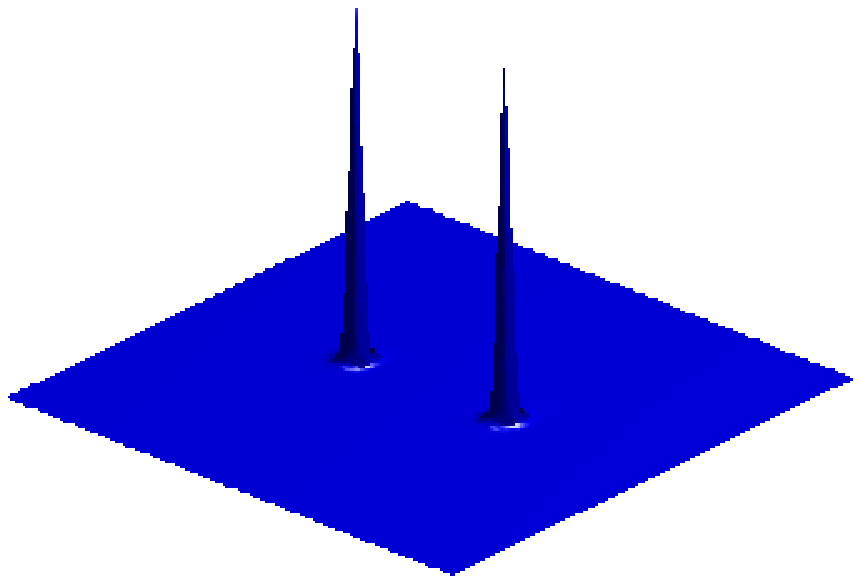}
\label{figd}}
\subfloat[Energy density I][$c=1$]{
\includegraphics[width=0.34\textwidth]{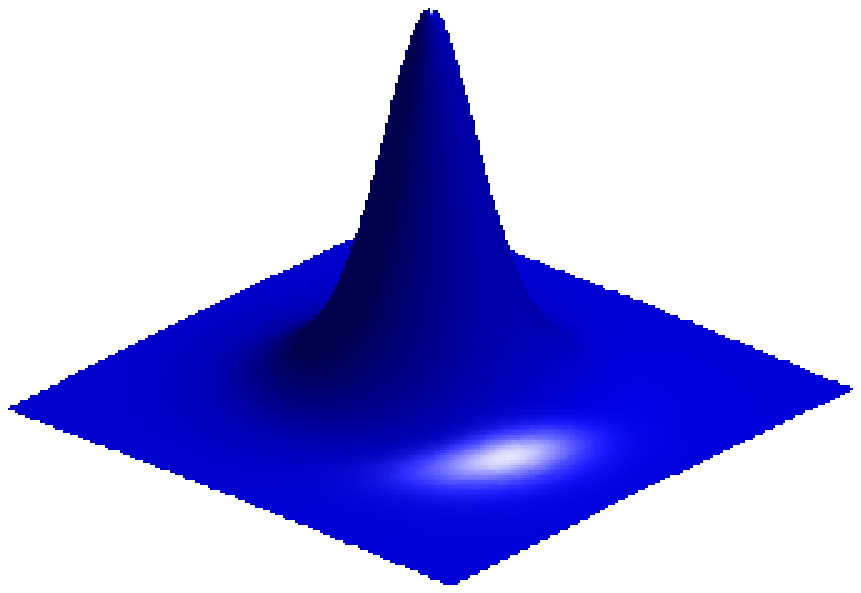}
\label{fige}}
\subfloat[Energy density II][$c=1.1$]{
\includegraphics[width=0.34\textwidth]{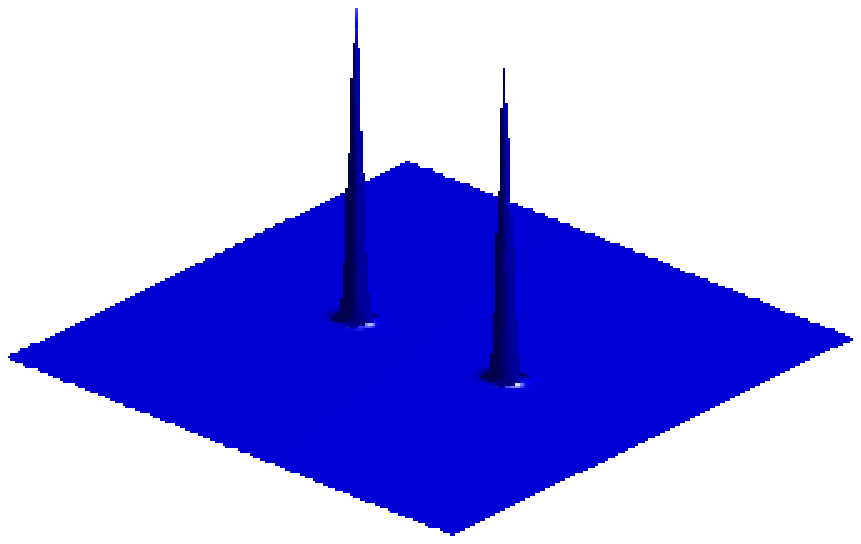}
\label{figf}}
\qquad
\subfloat[Energy density I][$c=2$]{
\includegraphics[width=0.3\textwidth]{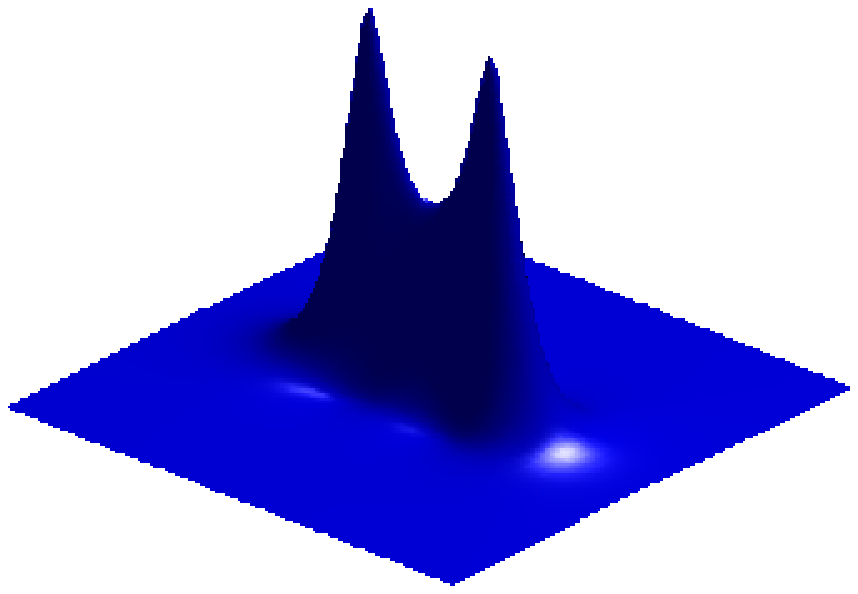}
\label{figg}}
\subfloat[Energy density I][$c=10$]{
\includegraphics[width=0.34\textwidth]{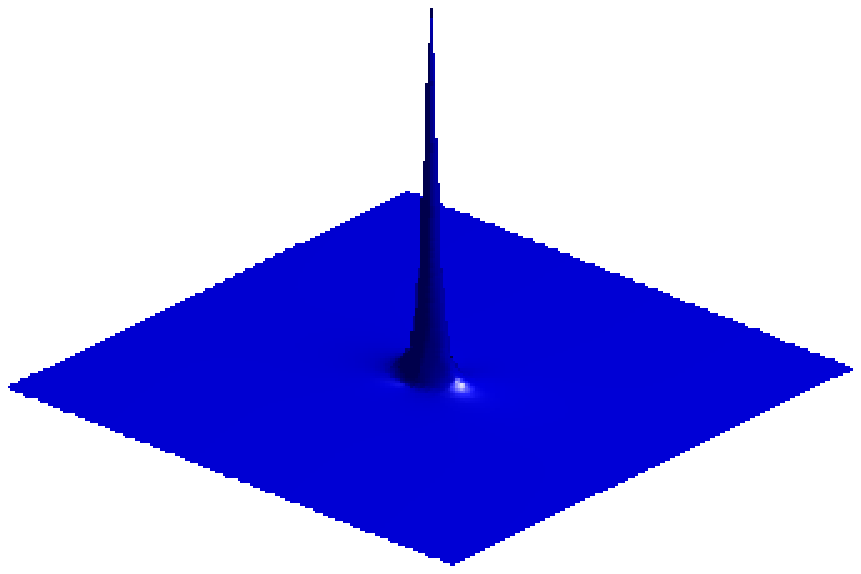}
\label{figh}}
\subfloat[Energy density II][$c=\infty$]{
\includegraphics[width=0.34\textwidth]{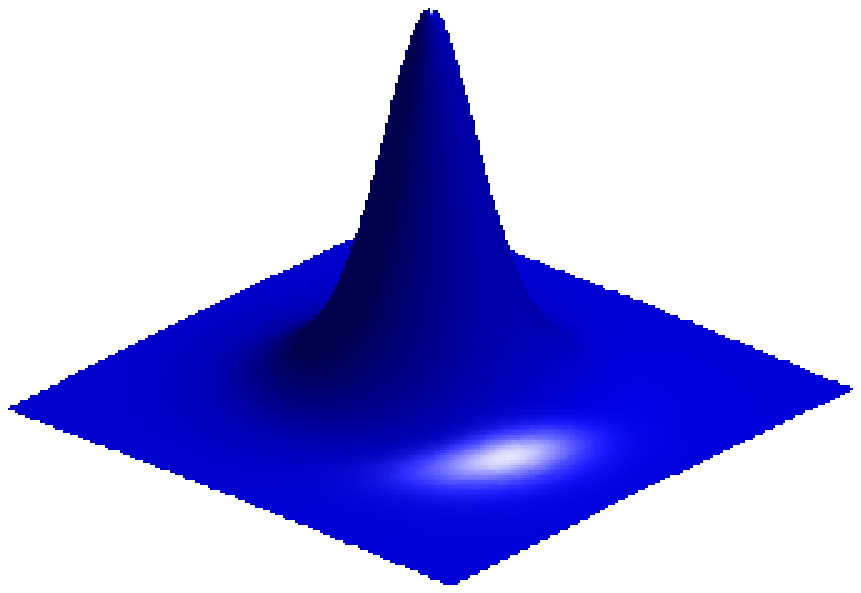}
\label{figi}}
\caption{These figures display the energy densities of charge three lumps given by $R_c$ in (\ref{Rc}) for different values of $c$.}
\label{enrdsyRc}
\end{figure}
 
\begin{proposition}
$\widetilde{Rat_3}$  is a non-compact totally geodesic Lagrangian submanifold of the space of degree $3$ lumps $Rat_3.$
 \end{proposition}
\begin{proof}
The proof can be found in Ref. \cite{speight2}.
\end{proof}
Our next step is to compute the metric on the moduli space of ${Rat}_3$.  
Consider the rotation group $SO(3)$ action on the coordinate systems $z$ and $R$, $R\in {Rat}_3$.
Consider first the action $z\mapsto Uz, U\in SO(3)\cong PU(2)\cong
SU(2)/{\mathbb{Z}_2}$. We can expand the left
invariant $1-$form $U^{-1}dU$ in terms of a convenient basis of the
Lie algebra $\frac{i}{2}\tau_a,$ $a=1,2,3$, where $\tau_a$ are Pauli
matrices:
\begin{equation}
\label{udu}
U^{-1}dU=\sigma\cdot\frac{i}{2}\tau_a=\sigma_1t_1+\sigma_2t_2+\sigma_3t_3,
\end{equation}
where
$d\sigma_i=\frac{1}{2}\varepsilon_{ijk}\sigma_j\wedge\sigma_k$.
Similarly, for the action $R\mapsto MR,M\in SO(3)$, we have an expression in the Lie algebra
$\frac{i}{2}\tau_a,a=1,2,3$:
\begin{equation}
\label{mdm}
M^{-1}dM=\eta\cdot\frac{i}{2}\tau_a=\eta_1t_1+\eta_2t_2+\eta_3t_3,
\end{equation}
where $d\eta_i=\frac{1}{2}\varepsilon_{ijk}\eta_j\wedge\eta_k$. For example, consider  $\tilde M\in SU(2)$ defined by  
\begin{equation*}
\tilde{M}=\begin{pmatrix}
e^{\frac{i}{2}(\psi+\phi)}\cos(\frac{\theta}{2})&e^{\frac{i}{2}(\psi-\phi)}\sin(\frac{\theta}{2})\\
-e^{\frac{i}{2}(\phi-\psi)}\sin(\frac{\theta}{2})&e^{-\frac{i}{2}(\psi+\phi)}\cos(\frac{\theta}{2})
\end{pmatrix}.
\end{equation*}
We can then see that $\tilde{M}^{-1}d\tilde{M}=\eta_1t_1+\eta_2t_2+\eta_3t_3$ where the $\eta_i$'s are computed as
\begin{align*}
\eta_1 &=-\sin\psi d\theta+\cos\psi\sin\theta d\phi\\
\eta_2&=\cos\psi d\theta+\sin\psi\sin\theta d\phi\\
\eta_3&=d\psi+\cos\theta d\phi.
\end{align*} 

Furthermore,  let $R\in SU(2)$ and $M\mapsto MR, M\in  SO(3)$. Then we can find that $\sigma\mapsto\mathcal{R}\sigma$ and $\eta\mapsto\mathcal{R}\eta$, where $\mathcal{R}\in SO(3)$ with matrix component $R_{ab}=\frac{1}{2} {\mathrm tr}(\tau_aR^{\dag}\tau_bR)$. Hence both $\sigma$ and $\eta$ transform as $3-$vectors under rotations. One can change from the coordinate basis on $SO(3)$, $\{d\alpha,d\beta,d\gamma\}$, to the left invariant $1$-forms on $SO(3)$ which  are given by (\ref{udu}) and (\ref{mdm}) as before, but with the range of angles  appropriate to $SO(3)$, $\alpha\in[0,\pi),\beta\in[0,2\pi),\gamma\in[0,2\pi)$. 

The metric is invariant under spatial rotations and rotations in target space. Then by considering $dc,\sigma$ and $\eta$, we can construct the most general metric as
\begin{small}
\begin{equation}
\label{Rcmetric}
g=
A_{ij}(c)\sigma_i\sigma_j+B_i(c)\sigma_idc+C(c)d^2c+D_{ij}(c)\eta_i\eta_j+E_i(c)\eta_idc
+F_{ij}(c)\sigma_i\eta_j,
\end{equation}
\end{small}
where $i,j=1,2,3.$ Furthermore, each component function depends only on $c$ and is independent of the Euler angles.

The transformations
$\rho:z\mapsto\bar{z}$ and $w:R\mapsto\bar{R}$  map lumps to anti-lumps because each reverses the sign of the
topological degree,  and so both are not an isometry of the moduli space. In fact, the composite transformation $w\circ\rho$ is an isometry. Consider
the isometry transformation $U\mapsto\bar{U}$, where $ U\in SO(3)$  as a $SU(2)$ M\"{o}bius transformation and suppose again $c\mapsto c$.
Then
\begin{align}
\label{sigmaeta}
&\sigma=(\sigma_1,\sigma_2,\sigma_3)\mapsto(-\sigma_1,\sigma_2,-\sigma_3)\, \, \,\,\text{and}\, \, \, \,
\eta=(\eta_1,\eta_2,\eta_3)\mapsto(-\eta_1,\eta_2,-\eta_3).
\end{align}
This isometry removes $B_i(c)$ and $E_i(c)$ for $i=1,3$ from the general possible metric equation (\ref{Rcmetric})
because for $c\mapsto c$, we have that
\begin{align*}
&\sigma\cdot dc\mapsto (-\sigma_1 dc,\sigma_2dc,-\sigma_3dc),\\
 &\eta\cdot dc\mapsto (-\eta_1dc,\eta_2dc,-\eta_3dc).
\end{align*}
The isometry (\ref{sigmaeta}) also  results in $A_{12}(c)\equiv A_{21}(c)\equiv A_{23}(c)\equiv A_{32}(c)\equiv 0$ and  $D_{12}(c)\equiv D_{21}(c)\equiv D_{23}(c)\equiv D_{32}(c)\equiv F_{12}(c)\equiv F_{21}(c) \equiv F_{13}(c)\equiv F_{31}(c)\equiv F_{23}(c)\equiv F_{32}(c)\equiv0$.  Furthermore, we can use the fact that $R\in \widetilde{Rat_3}$ has $D_2$ symmetry. Take a $\pi$ rotation around the third axis
\begin{align}
\label{sigmaeta1}
&(\sigma_1, \sigma_2, \sigma_3) \mapsto (-\sigma_1, -\sigma_2, \sigma_3)\, \, \, \, \text{and}\, \, \, \,
(\eta_1, \eta_2, \eta_3) \mapsto (-\eta_1, -\eta_2, \eta_3).
\end{align}
This isometry (\ref{sigmaeta1})  gives ${A}_{13}(c)\equiv {A}_{31}(c)\equiv {D}_{13}(c)\equiv {D}_{31}(c)\equiv{F}_{13}(c)\equiv {F}_{31}(c)\equiv  {B}_2(c)\equiv{E}_2(c)\equiv0$
because  we have that
$$
\sigma_1\sigma_3\mapsto -\sigma_1 \sigma_3, \quad  
\eta_1\eta_3\mapsto -\eta_1 \eta_3, \quad
\sigma_2 dc\mapsto-\sigma_2 dc\quad \text{and}\quad 
\eta_2dc\mapsto-\eta_2 dc.
$$ 
Hence, they can be removed from the general possible metric equation (\ref{Rcmetric}).

Our next task is finding the remaining metric functions of $c$ by taking the
appropriate Euler angles. Firstly, consider the parametrisation of $SO(3)$ by
\begin{equation}
\label{matrix}
M(\alpha,\theta,\varphi)=
\begin{pmatrix}
\cos\frac{\alpha}{2}+i\sin\frac{\alpha}{2}\cos\theta &
i\sin\frac{\alpha}{2}(\cos\varphi+i\sin\varphi)\sin\theta\\
i\sin\frac{\alpha}{2}(\cos\varphi-i\sin\varphi)\sin\theta &
\cos\frac{\alpha}{2}-i\sin\frac{\alpha}{2}\cos\theta
\end{pmatrix}.
\end{equation}
Take first the action $R\mapsto R_{\star}=e^{i\alpha}R$. That is, we consider $\theta=0$ in $M(\alpha,\theta,\varphi)$.  Then we
have a metric  of the form
$\gamma=\gamma_{\alpha\alpha}(c)d^2\alpha$, where
\begin{align*}
D_{33}(c)=\gamma_{\alpha\alpha}(c)&=\int_D\frac{|\partial_{\alpha}R_{\star}|^2}{(1+|R_{\star}|^2)^2}
\frac{dzd\bar{z}}{(1+|z|^2)^2}
=\int_D\frac{|R|^2}{(1+|R|^2)^2}
\frac{dzd\bar{z}}{(1+|z|^2)^2}.
\end{align*}
Now we can also evaluate $D_{11}(c)$ and $D_{22}(c)$. Suppose we are taking $\theta=\frac{\pi}{2}$ and $\varphi=0$ in $R\mapsto R_{\star}=MR$, where $M$ is
given by the matrix (\ref{matrix}). Then,  we have a metric of the form $\gamma=\gamma_{\alpha\alpha}(c)d^2\alpha$ where
\begin{align*}
D_{11}(c)
=\gamma_{\alpha\alpha}(c)&=\int_D\frac{|\partial_{\alpha}R_{\star}|^2}{(1+|R_{\star}|^2)^2}\frac{dzd\bar{z}}{(1+|z|^2)^2}
=\frac{1}{4}\int_D\frac{|1-R^2|^2}{(1+|R|^2)^2}\frac{dzd\bar{z}}{(1+|z|^2)^2}.
\end{align*}
Similarly, to find the expression for $D_{22}(c),$  take
$\alpha=\frac{\pi}{2}$ and $\varphi=0$ in our parametrisation of
$SO(3)$ in (\ref{matrix}).
Then the metric is of the form
$\gamma=\gamma_{\theta\theta}(c)d^2\theta$ where
\begin{align*}
D_{22}(c)=\gamma_{\theta\theta}(c)=
&\int_D\frac{|\partial_{\theta}R_{\star}|^2}{(1+|R_{\star}|^2)^2}\frac{dzd\bar{z}}{(1+|z|^2)^2}
=\frac{1}{4}\int_D\frac{|1+R^2|^2}{(1+|R|^2)^2}\frac{dzd\bar{z}}{(1+|z|^2)^2}.
\end{align*}

To find the functions $A_{ii}(c), i=1,2,3$, we can follow the
same argument in evaluating the $D_{ii}(c), i=1,2,3$. Suppose  we are
considering the same parametrisation of $SO(3)$ as (\ref{matrix}) and the
$SO(3)$ action on $z$. For instance, let $z\mapsto e^{i\alpha}z$.
Then $R\mapsto R_{\star}=R(e^{i\alpha}z)$. Therefore, we are now
able to find $A_{33}(c)$ from the metric of the form
$\gamma=\gamma_{\alpha\alpha}(c)d^2\alpha$ where
\begin{align*}
A_{33}(c)=\gamma_{\alpha\alpha}(c)
&=\int_D\frac{|\partial_{\alpha}R_{\star}|^2}{(1+|R_{\star}|^2)^2}\frac{dzd\bar{z}}{(1+|z|^2)^2}
=\int_D\frac{|z|^2|\frac{dR}{dz}|^2}{(1+|R|^2)^2}\frac{dzd\bar{z}}{(1+|z|^2)^2}.
\end{align*}
We can also find the other two functions $A_{11}(c)$ and $A_{22}(c)$ by
taking the Euler angles ($\theta=\frac{\pi}{2}$ and $\varphi=0$) and ($\alpha=\frac{\pi}{2}$ and $\varphi=0$), respectively.
Considering the above Euler angles and from the metrics of the form $\gamma=\gamma_{\alpha\alpha}(c)d^2\alpha$ and $\gamma=\gamma_{\theta\theta}(c)d^2\theta$, we can find that
\begin{equation*} A_{11}(c)=\gamma_{\alpha\alpha}(c)=\frac{1}{4}\int_D\frac{|1-z^2|^2|\frac{dR}{dz}|^2}{(1+|R|^2)^2}\frac{dzd\bar{z}}{(1+|z|^2)^2}
\end{equation*}
and
\begin{equation*} A_{22}(c)=\gamma_{\theta\theta}(c)=\frac{1}{4}\int_D\frac{|1+z^2|^2|\frac{dR}{dz}|^2}{(1+|R|^2)^2}\frac{dzd\bar{z}}{(1+|z|^2)^2}.
\end{equation*}
 Similarly, we can find the following metric functions  $F_{ii}(c), i=1,2,3$ as
\begin{align*}
&F_{11}(c)=\frac{1}{4}\int_D\frac{\Re{\left((1-z^2)(1-\bar{R}^2)\frac{dR}{dz}\right)}}{(1+|R|^2)^2}\frac{dzd\bar{z}}{(1+|z|^2)^2},\\
&F_{22}(c)=\frac{1}{4}\int_D\frac{\Re{\left((1+z^2)(1+\bar{R}^2)\frac{dR}{dz}\right)}}{(1+|R|^2)^2}\frac{dzd\bar{z}}{(1+|z|^2)^2},\\
&F_{33}(c)=\int_D\frac{\Re{\left(z\bar{R}\frac{dR}{dz}\right)}}{(1+|R|^2)^2}\frac{dzd\bar{z}}{(1+|z|^2)^2}.
\end{align*}
These expressions have also been derived in a different context in \cite{Manko:2007pr}.
Finally, the function $C(c)$ is given by
\begin{equation*}
C(c)
=\int_D\frac{|\partial_cR|^2}{(1+|R|^2)^2}\frac{dzd\bar{z}}{(1+|z|^2)^2}.
\end{equation*}
It can be shown that $C(c)$ is bounded from above by $2\pi,$ and this can also been seen in Fig.~\ref{mfigd}.
Hence, the metric on the $7$-dimensional space of charge three lumps is given by
\begin{align}
\label{gmetric}
g=& C(c)d^2c + \sum\limits_{i=1}^3 \left(A_{ii}(c)\sigma_i^2+D_{ii}(c)\eta_i^2+F_{ii}(c)\sigma_i\eta_i\right).
\end{align}
The coefficient functions of the metric (\ref{gmetric}) are displayed in 
Fig.~\ref{metricRa}. As $c=0$ is the axial map, we find that $A_{11}(0)=A_{22}(0)$, $D_{11}(0)=D_{22}(0)$ and $F_{11}(0)=F_{22}(0)=0$.

\begin{figure}[ht!]
\centering
\subfloat[Energy density I][]{
\includegraphics[width=0.5\textwidth]{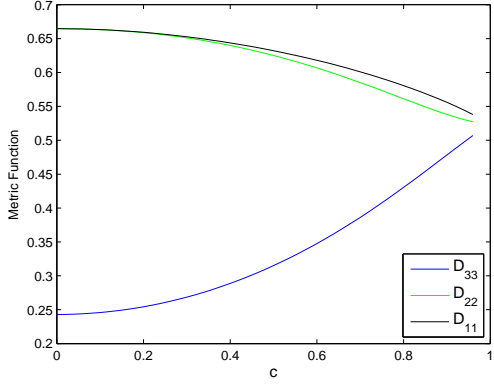}
\label{mfiga}}
\subfloat[Energy density I][]{
\includegraphics[width=0.5\textwidth]{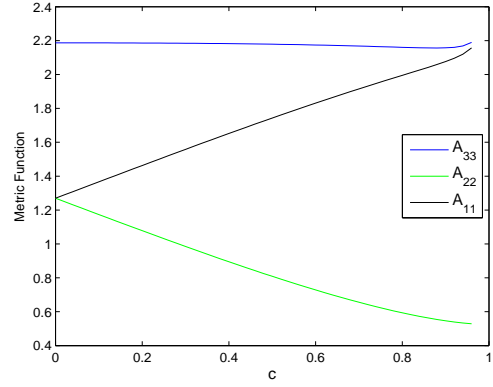}
\label{mfigb}}
\qquad
\subfloat[Energy density I][]{
\includegraphics[width=0.5\textwidth]{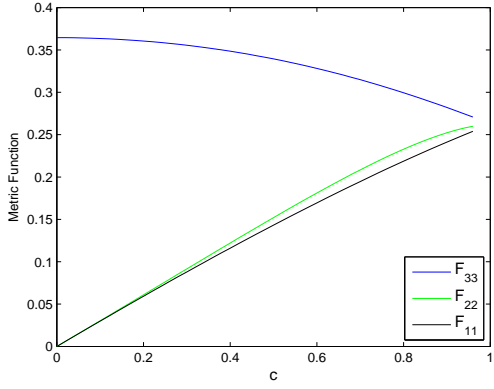}
\label{mfigc}}
\subfloat[Energy density I][]{
\includegraphics[width=0.5\textwidth]{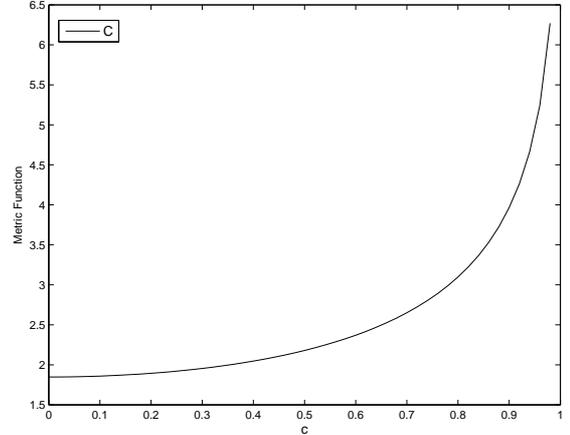}
\label{mfigd}}
\caption{(a) The metric functions $D_{11}(c)$, $D_{22}(c)$, $ D_{33}(c)$. \,  
(b) The metric functions $A_{11}(c)$,\, $A_{22}(c)$,\, $A_{33}(c)$. \,  
(c) The metric functions $F_{11}(c),\, F_{22}(c)$,\, $F_{33}(c)$. \,  (d) The metric function $C(c)$.  
\label{metricRa}}
\end{figure}

\begin{proposition}
\label{finitevolume}
The moduli space $\widetilde{Rat_3}$ with respect to the metric $g$ in \eqref{gmetric} has finite volume.
\end{proposition}

\begin{proof}
The volume on the moduli space $\widetilde{Rat_3}$ is given by 
\begin{equation}
Vol(\widetilde{Rat_3})= \frac{1}{4} \int_{SO(3)\times SO(3)\times\mathbb{R}} dVol_g(\widetilde{Rat_3}),
\end{equation}
where the factor of $\frac{1}{4}$ arises to avoid overcounting configurations that are related by $D_2$ symmetry, and
$dVol_g(\widetilde{Rat_3})$ is the volume element on $(\widetilde{Rat_3}, g)$ given by
\begin{equation}
\label{volume}
dVol_g(\widetilde{Rat_3})=
\sqrt{|\det(g_{ij})|}\sigma_1\wedge\sigma_2\wedge\sigma_3 \wedge\eta_1\wedge\eta_2\wedge\eta_3\wedge dc.
\end{equation}
Note that $Vol(SO(3))=\int_{SO(3)}\sigma_1\wedge\sigma_2\wedge\sigma_3=8\pi^2$ and similarly  $\int_{SO(3)}\eta_1\wedge\eta_2\wedge\eta_3=8\pi^2$.
The determinant of $g_{ij}$ is calculated as
\begin{align}
\label{detmetric}
\det(g_{ij})& = 
(A_{11}A_{22}A_{33}D_{11}D_{22}D_{33}C)(c).
\end{align}
We can also find the following inequalities for the metric functions which are $$
A_{22}\le\frac{4}{(1+c^2)^2},
$$ 
and both  $A_{11}$ and $A_{33}$ are bounded above by $\frac{2\pi}{3}+\frac{2\pi^2\sqrt{3}}{243}$. Similarly, we can see that $D_{33}\le\frac{\pi}{6}$, and $D_{11}$ and $D_{22}$ are both bounded above by
$\frac{\pi}{6}+\frac{2\pi^2\sqrt{3}}{243}$. Hence, the determinant can be bounded via
$$
\det(g_{ij})\le K \frac{4}{(1+c^2)^2},
$$ 
where 
$$
K =  \frac{\pi^2}{3} \left(\frac{2\pi}{3}+\frac{2\pi^2\sqrt{3}}{243}\right)^2
\left(\frac{\pi}{6}+\frac{2\pi^2\sqrt{3}}{243}\right)^2.  
$$
Note that $\int_0^1\frac{2}{1+c^2}dc=\frac{\pi}{2}$ which implies the following integral 
$$
Vol(\widetilde{Rat_3}) = \frac{\left(8 \pi^2\right)^2}{4}\int_0^{1}\sqrt{|\det(g_{ij})|}\,dc
\le 8\pi^5 \sqrt{K}.
$$
This proves the volume of the moduli space of charge three lumps is finite. 
\end{proof}

\begin{proposition} The moduli space 
$(\widetilde{Rat_3},g)$ has a geodesic submanifold of finite length. 
\end{proposition}
\begin{proof}
Consider the submanifold of $\widetilde{Rat_3}$ given by $R_c$ in \eqref{Rc} for $0 \le c < 1.$ As $\gamma_{cc} = C(c)$ is bounded then 
$$
\int\limits_0^1\sqrt{\gamma_{cc}}\,dc<\infty,
$$ 
and hence the length is finite. Therefore, the boundary of $(\widetilde{Rat_3},g)$ at $c=1$ lies at finite distance from the axial map $(c=0),$ so that the space $\widetilde{Rat_3}$ is geodesically incomplete.
\end{proof}

Finally, we consider  rational functions of the form $R(z)=z^3$ and its symmetry orbit denoted by  $\widetilde{Rat_3^0}$.  The energy density is symmetric and its metric is equivalent to the metric on the moduli space $\{\xi z^n:\xi\in\mathbb{C}^{\times}\}$ with $\xi=1$ and $n=3$ which was shown in Ref. \cite{speight4}. Note that $\widetilde{Rat_3^0}$ is a totally geodesic submanifold of  $\widetilde{Rat_3}.$ 
The general metric $g_0$ on $\widetilde{Rat_3^0}$ is given by
\begin{align}
\label{ginf}
g_0&=f_i\sigma_i^2+h_i\eta_i^2+F_{33}\sigma_3\eta_3,\\
&=f_1\sigma_1^2 +f_2\sigma_2^2 +f_3(\sigma_3+3\eta_3)^2+h_1\eta_1^2+h_2\eta_2^2,
\end{align}
where
\begin{align*}
&f_1=f_2=\frac{2\pi^2\sqrt{3}}{27},\, \, \, h_1=h_2=\frac{\pi}{6}+\frac{2\pi^2\sqrt{3}}{243},\\
&f_3=\frac{3\pi}{2}-\frac{4\pi^2\sqrt{3}}{27},\, \, \,h_3=\frac{\pi}{6}-\frac{4\pi^2\sqrt{3}}{243}, \, \, \, \,
F_{33}=\frac{\pi}{2}-\frac{4\sqrt{3}\pi^2}{81}.
\end{align*}

\begin{proposition}
$\widetilde{Rat_3^0}$  has a finite volume with \\
 $Vol(\widetilde{Rat_3^0})=\frac{\pi^{13/2}}{19683}\left( 4\,\pi \,\sqrt {3}+81 \right) \sqrt {324-32\,\sqrt {3}}
$.
\end{proposition}
\begin{proof}
Here to avoid over-counting, we divide the volume element by two since the space has an additional $C_2$ symmetry.  As earlier in Proposition \ref{finitevolume}, we have that $Vol(SO(3))=\int_{SO(3)}\sigma_1\wedge\sigma_2\wedge\sigma_3=8\pi^2$. From matrix (\ref{matrix}),  $\eta_1\wedge\eta_2=\sin(\theta)d\phi d\theta, \, 0\le\theta\le\pi, \, 0\le\phi\le2\pi$, the integral is $\int_{SO(3)/SO(2)}\eta_1\wedge\eta_2=4\pi$. 
The volume is given by
\begin{equation*}
Vol(\widetilde{Rat_3^0})=\int\sqrt{|f_1f_2f_3h_1h_2|}\, \,\eta_1\wedge\eta_2\wedge\sigma_1\wedge\sigma_2\wedge\sigma_3,
\end{equation*}
which can be evaluated as 
\begin{align*}
Vol(\widetilde{Rat_3^0})&=32\pi^3\sqrt{|f_1f_2f_3h_1h_2|}
=\frac{4\pi^{13/2}}{19683}\left( 4\,\pi \,\sqrt {3}+81 \right) \sqrt {324-32\,\sqrt {3}}.
\end{align*}
\end{proof}

\section{Conclusion}

In this paper, we studied the moduli space of lumps on real projective space $\widetilde{Rat_N}$ which is given by the space of rational maps of degree $N$ subject to symmetry requirements \cite{speight2}. We examined this moduli space using two different approaches. First we classified all possible cyclic symmetries of $\widetilde{Rat_N}.$ Then we analysed $\widetilde{Rat_N}$ using the Riemann-Hurwitz formula.

The symmetry requirements for $\widetilde{Rat_N},$ which from  a geometric point of view mean that zeros and poles have to be opposite, greatly reduce the number of allowed cyclic symmetries $C_n^k$ compared to general rational maps. The allowed symmetries are given in Lemma \ref{lem2}, and can be compared to Lemma \ref{lem1}, which states a similar results for general rational maps \cite{kruschfinkelstein}.

We then focused on the case $N=3.$ Imposing a $C_2$ symmetry automatically leads to a $D_2$ symmetric map, and a cyclic symmetry $C_n$ for $n\ge 3$ results in an axially symmetric rational map with $D_\infty$ symmetry. On $\widetilde{Rat_N},$ the symmetry group $SO(3) \times SO(3)$ of rotations and isorotations acts isometrically. Hence, imposing $C_3$ and $C_2$ results in 
a $5$-dimensional orbit of the axial map and a $7$-dimensional family of orbits with dihedral symmetry $D_2,$ respectively. Note that the moduli space of charge one lumps on the two-sphere can also be parametrised as families of orbits depending on one parameter, where the orbits are generically five-dimension due to an axial symmetry \cite{speight2}. Hence, the moduli space of charge one lumps on the two-sphere is $6$-dimensional and provided the motivation for our work. The Riemann-Hurwitz formula also decomposes $\widetilde{Rat_3}$ into a $5$-dimensional space of maps with two ramification points of index $3$ and a $7$-dimensional space with $4$ ramification points of index $2,$ and we showed that these two points of view produce the same spaces. In fact, we constructed explicitly the M\"obius transformations in space and target space that are needed to relate the two $7$-dimensional spaces. In summary, the moduli space of charge three lumps is a $7$-dimensional manifold which can be described as a family of symmetry orbits of $D_2$ symmetric maps $R_c$ in (\ref{Rc}) where $c=0$ denotes the axially symmetric map. Care has to be taken to avoid overcounting since the $SO(3)\times SO(3)$ action produces configurations that are related by $D_2$ symmetry. 
Furthermore, as $c \to 1$ or $c \to \infty$ two lumps become increasingly spiky and collapse, as two poles of $R_c$ cancel with two zeros. The symmetry requirement that poles and zeros have to be opposite results in a more complicated lump decay, and in particular, $D_\infty$ symmetry prevents lump decay.

The dihedral symmetry of $\widetilde{Rat_3}$ allowed us to find explicit expressions of the $L^2$ metric \cite{speight2} that is induced on the moduli space by the kinetic energy. It is rare that the moduli space metric can be evaluated for topological charge greater than one. Recently, the metric of hyperbolic vortices of charge two has been calculated in Ref. \cite{Lamia}
generalising the formula obtained in \cite{strachan}.
We showed that the volume of $\widetilde{Rat_3}$ is finite. Furthermore, we constructed a geodesic that connects the axial map to the boundary of $\widetilde{Rat_3}$ which has finite length. This shows that $\widetilde{Rat_3}$ is geodesically incomplete, as shown in Ref. \cite{speight2}. We also evaluated the volume of the space of axially symmetric maps.

By imposing the cyclic symmetry $C_n^k$ for $0<k<n$ we identified some interesting $7$-dimensional submanifolds with dihedral symmetry $D_n,$ which have the same structure as the moduli space of $N=3$ lumps. These spaces are very similar to the moduli spaces of vortex polygons in \cite{kruschspeight2} and to the cyclically symmetric scattering of monopoles in \cite{Hitchin:1995qw}. It is possible to evaluate the metric of submanifolds with $D_n$ symmetry in a similar way to the metric of $\widetilde{Rat_3}.$
When $N=5$ the symmetry approach gives the following possible symmetries: 
$D_\infty,$ $O,$ $D_4,$ $D_3,$ $D_2$ and $C_2$ and no symmetry. The Riemann-Hurwitz formula provides an alternative decomposition into $5$ different spaces, but now there is no obvious correspondence between the two approaches apart from the axially symmetric case. How these spaces are related is an interesting problem for further study.

\subsection*{Acknowledgment}

SK would like to thank Martin Speight for suggesting to look into this problem and further useful discussions. SK is grateful to Jim Shank for stimulating discussions with regards to the geometry of rational maps. SK also would like to thank Nick Manton for an interesting discussion about zeros of Wronskians and their relation to monopole moduli spaces. SK acknowledges the EPSRC for the grant EP/I034491/1. AM would like to thank Mareike Haberichter for assisting with Matlab plots  and useful discussions. AM acknowledges SMSAS at the University of Kent for a PhD studentship.

\appendix
\section{Relationship between the rational maps $R_a(z)$ and $R_c(z)$}
\label{Appendix}

In this appendix, we show in detail how to transform the rational map 
$R_c(z)$ in (\ref{Rc}) into $R_a(z)$ in (\ref{Ra}) by constructing the relevant M\"obius transformations in space and target space. Thereby, we prove explicitly that these two maps are in the same orbit of the symmetry group, and hence both generate the moduli space $\widetilde{Rat_3}.$
As a starting point, we calculate the zeros of $dR/dz,$ that is the zeros of the Wronskian $W = p^\prime q - p q^\prime,$ where $R = p/q.$ These zeros are invariant under rotations in target space. For 
$R_a(z),$ the zeros of the Wronskian $W_a(z)$ are at $0,$ $\infty$ and
$$ 
z_\pm = \frac{a^2-3\pm \sqrt{a^4+10a^2+9}}{4a}.
$$
Note that as $a$ is increased from $0$ to $\infty$, the zero $z_+$ moves along the positive real axis from $0$ to $\infty$ whereas the zero $z_-$ moves along the negative real axis from $-\infty$ to $0.$

Before calculating the Wronskian $W_c$ for $R_c,$ it is convenient to make a change of variables. For $c \in (0,1)$ we define
\begin{equation}
\label{u(c)}
u=\frac{3}{2c}-\frac{c}{2},
\end{equation}
which is a monotonously decreasing function of $c,$ taking values in $(1,\infty).$ Hence this map is a bijection between $(0,1)$ and $(1,\infty)$ and has the inverse
\begin{equation}
\label{c(u)}
c =  -u+\sqrt{u^2+3}.
\end{equation}
Then the zeros of $W_c$ are given in term of $u$ by
\begin{eqnarray*}
\label{zerosc}
z_1 = i \sqrt{u - \sqrt{u^2-1}}, &\quad& z_2 = -i \sqrt{u - \sqrt{u^2-1}}, \\
z_3 = i \sqrt{u + \sqrt{u^2-1}}, &\quad& z_4 = -i \sqrt{u + \sqrt{u^2-1}}.
\end{eqnarray*}
Since $u>1$ the zeros (\ref{zerosc}) are purely imaginary. For $u\to \infty$ the zeros $z_1$ and $z_2$ tend to $0,$ whereas the zeros $z_3$ and $z_4$ tend to $\pm i \infty.$ Hence, as $u \to \infty$ the corresponding rational map tends to the axial map. We now move the zero $z_1$ to $0$ via the M\"obius transformation 
$$
M_1(z) = \frac{z  + z_1 }{z_1 z + 1},
$$ 
which corresponds to a rotation around the $x_1$ axis in space. This moves the zeros of the Wronskian to $0,$ $\infty$ and to two imaginary zeros. Hence, we rotate by $\pi/2$ around the $x_3$ axis, 
$$
M_2(z) = -i z,
$$
which leaves the zeros of the Wronskian at $0$ and $\infty$ fixed but moves the remaining to zeros $Z_\pm$ onto the real axis. The equation for the zeros of the corresponding Wronskian can be simplified to
$$
z\left(z^2 + \frac{u-3}{\sqrt{2(u-1)}}z-1\right)=0,
$$
which has zeros at $0,$ $\infty$ and
$$
Z_+= \sqrt{\frac{2}{u-1}}, \quad \mathrm{and} \quad 
Z_-=-\sqrt{\frac{u-1}{2}}.
$$
Note that $Z_+$ travels along the real axis from $0$ to $\infty$ as $u$ decreases. By equating the two positive roots $Z_+$ and $z_+$ we can express
$a$ in terms of $u$ as
\begin{equation}
\label{a(u)}
a = \frac{3-u+\sqrt{u^2+3}}{\sqrt{2(u-1)}} 
\end{equation}
or alternatively,
\begin{equation}
\label{a(c)}
a = \sqrt{\frac{c(c+3)}{1-c}}.
\end{equation}
This relates the parameters $a$ and $c$ of the rational map, and as expected, $a$ moves along the real axis from $0$ to $\infty$ as $c$ increase from $0$ to $1$ (or $u$ decreases from $\infty$ to $1$).

In order to relate the rational maps $R_c\left(M_2\left(M_1\left(z\right)\right)\right)$ and $R_a(z)$ we need to perform a M\"obius transformation in target space. We do this again in two steps. It can be shown that $R_{c,0} = R_c\left(M_2\left(M_1\left(0\right)\right)\right)$ is purely imaginary. Therefore, the M\"obius transformation ${\tilde M}_1$ given by 
$$
{\tilde M}_1 \left(R_c\left(M_2\left(M_1\left(z\right)\right)\right)\right)
=\frac{R_c\left(M_2\left(M_1\left(z\right)\right)\right) - R_{c,0}}{-R_{c,0} R_c
\left(M_2\left(M_1\left(z\right)\right)\right) +1}
$$
is a rotation around the $x_1$ axis in target space which sets the image of the critical point at $z=0$ to $0,$ namely  ${\tilde M}_1 \left(R_c\left(M_2(M_1(0)\right)\right) =0.$ However, now the value of this map is purely imaginary for $0<c<1$ and real $z.$ So, we need the further M\"obius transformation ${\tilde M}_2$ given by 
$$
{\tilde M}_2 \left({\tilde M}_1 \left(R_c\left(M_2\left(M_1\left(z\right)\right)\right)\right)\right) 
= -i  \left({\tilde M}_1 \left(R_c\left(M_2\left(M_1\left(z\right)\right)\right)\right)\right)
$$  
which is a $\pi/2$ rotation around the $x_3$ axis in target space. This brings the map in the form of equation (\ref{Ra}) where $a$ is given in terms of $c$ by equation \eqref{a(c)}.


\begin{thebibliography}{100} 

\bibitem{Lamia}
  L.~S.~Alqahtani and J.~M.~Speight,
  ``Ricci magnetic geodesic motion of vortices and lumps,''
  arXiv:1410.4698 [math-ph], {\it to appear in the Journal of Geometry and Physics.}
%  doi:10.1016/j.geomphys.2015.07.008

\bibitem{bott}
R.~ Bott and L.~ W. ~Tu,
 ``Differential forms in algebraic topology,''
Springer-Verlag, 1982.

\bibitem{cova}
 R.~J.~Cova and W.~J.~Zakrzewski,
  ``Scattering of periodic solitons,''
  Rev.\ Mex.\ Fis.\  {\bf 50} (2005) 527.
%  [hep-th/0308161].

\bibitem{eells} 
J.~Eells and L.~Lemaire, 
 ``On the Construction of Harmonic and Holomorphic Maps between Surfaces,''  
Math.\ Ann.\ {\bf 252} (1980) 27.

\bibitem{fulton1}
W.~Fulton and R.~Weiss,
  ``Algebraic Curves: an introduction to algebraic geometry,''  volume 30. Benjamin New York, 1969.

\bibitem{Hitchin:1995qw}
  N.~J.~Hitchin, N.~S.~Manton and M.~K.~Murray,
  ``Symmetric monopoles,''
  Nonlinearity {\bf 8} (1995) 661.
%  [dg-ga/9503016].
%%CITATION = DG-GA/9503016;%%

\bibitem{houghton}
  C.~J.~Houghton, N.~S.~Manton and P.~M.~Sutcliffe,
  ``Rational maps, monopoles and Skyrmions,''  Nucl.\ Phys.\ B {\bf 510} (1998) 507.
%[hep-th/9705151].

\bibitem{kruschfinkelstein}
S.~Krusch,
``Finkelstein-Rubinstein constraints for the Skyrme model with pion masses,''  Proc.\ Roy.\ Soc.\ Lond.\ A {\bf 462} (2006) 2001. 
%[hep-th/0509094].
%%CITATION = HEP-TH/0509094;%%

\bibitem{krusch-rat}
  S.~Krusch,
  ``Homotopy of rational maps and the quantization of skyrmions,''
  Annals Phys.\  {\bf 304} (2003) 103.
%[hep-th/0210310].
%%CITATION = HEP-TH/0210310;%%  

\bibitem{kruschspeight}
  S.~Krusch and J.~M.~Speight,
  ``Quantum lump dynamics on the two-sphere,''
  Commun.\ Math.\ Phys.\  {\bf 322} (2013) 95.
%%CITATION = ARXIV:1204.6203;%%

\bibitem{kruschspeight2}
  S.~Krusch and J.~M.~Speight,
  ``Exact moduli space metrics for hyperbolic vortices,''
  J.\ Math.\ Phys.\  {\bf 51} (2010) 022304.
 %[arXiv:0906.2007 [hep-th]].
 %%CITATION = ARXIV:0906.2007;%%

\bibitem{leese}
R.~A.~Leese,
``Low-energy Scattering of Solitons in the {CP}**1 Model,''  Nucl.\ Phys.\ B {\bf 344} (1990) 33.

\bibitem{Manko:2007pr}
  O.~V.~Manko, N.~S.~Manton and S.~W.~Wood,
  ``Light nuclei as quantized skyrmions,''
  Phys.\ Rev.\ C {\bf 76} (2007) 055203.
%  [arXiv:0707.0868 [hep-th]].
%%CITATION = ARXIV:0707.0868;%%

\bibitem{manton}
N.~S.~Manton and P.~Sutcliffe,
Topological solitons,  Cambridge University Press, 2004.

\bibitem{massey} 
W.~S.~Massey, A Basic Course in Algebraic Toplogy, Springer-Verlag, New York, 1991.

\bibitem{speight4}
  J.~A.~McGlade and J.~M.~Speight,
``Slow equivariant lump dynamics on the two sphere,''
 Nonlinearity {\bf 19} (2006) 441.

\bibitem{Muhamed}
A.~A.~Muhamed, ``Moduli Spaces of Topological Solitons'',
PhD thesis, University of Kent, 2015.

\bibitem{polyakov}
  A.~M.~Polyakov and A.~A.~Belavin,
  ``Metastable States of Two-Dimensional Isotropic Ferromagnets,''  JETP Lett.\  {\bf 22} (1975) 245   [Pisma Zh.\ Eksp.\ Teor.\ Fiz.\  {\bf 22} (1975) 503].

\bibitem{ruback}
  P.~J.~Ruback,
  ``$\sigma$ Model Solitons and Their Moduli Space Metrics,''
  Commun.\ Math.\ Phys.\  {\bf 116} (1988) 645.  
  
\bibitem{sadun}
  L.~A.~Sadun and J.~M.~Speight,
  ``Geodesic incompleteness in the $\mathbb{C}P^1$ model on a compact Riemann surface,''
  Lett.\ Math.\ Phys.\  {\bf 43} (1998) 329.
%  [hep-th/9707169].

\bibitem{samols}
  T.~M.~Samols,
  ``Vortex scattering,''  Commun.\ Math.\ Phys.\  {\bf 145} (1992) 149.

\bibitem{segal} 
G.~ Segal,
 ``The Topology of Spaces of Rational Functions,''  Acta. Math. {\bf 143} (1979) 39.

\bibitem{skyrme}
T.~H.~R. ~Skyrme, ``A nonlinear field theory, '' Proc.Roy. Soc. A260 127.

\bibitem{speight2}
  J.~M.~Speight,
  ``The $L^2$ geometry of spaces of harmonic maps $S^2\to S^2$  and $\mathbb{R}P^2\to\mathbb{R}P^2$,''
  J.\ Geom.\ Phys.\ {\bf 47} (2003) 343. 

\bibitem{speightlump}
  J.~M.~Speight,
  ``Lump dynamics in the $\mathbb CP^1$  model on the torus,''  Commun.\ Math.\ Phys.\  {\bf 194} (1998) 513. 
% [hep-th/9707101].

\bibitem{speight1} 
  J.~M.~Speight,
  ``Low-energy dynamics of a $\mathbb CP^1$ lump on the sphere,''
  J.\ Math.\ Phys.\  {\bf 36} (1995) 796.
%  [hep-th/9712089].

\bibitem{strachan}
  I.~A.~B.~Strachan,
  ``Low velocity scattering of vortices in a modified Abelian Higgs model,''
  J.\ Math.\ Phys.\  {\bf 33} (1992) 102.

\bibitem{ward}
  R.~S.~Ward,
  ``Slowly Moving Lumps in the $\mathbb CP^1$ Model in (2+1)-dimensions,''  Phys.\ Lett.\ B {\bf 158} (1985) 424.

\end{thebibliography}
\end{document}